\documentclass[11pt]{article}

\usepackage[margin=1in]{geometry}
\usepackage{amsmath,amssymb,amsthm}
\usepackage{graphicx}
\usepackage{booktabs}
\usepackage{microtype}
\usepackage[hidelinks]{hyperref}
\usepackage{tikz}
\usepackage{caption}
\usepackage{subcaption}
\usepackage{algorithm}
\usepackage{algpseudocode}
\usepackage{float}
\usepackage{url}
\usepackage{enumitem}
\usepackage{natbib}
\usetikzlibrary{arrows.meta,positioning,shapes.geometric,calc}

\definecolor{polycol}{RGB}{70,110,180}
\definecolor{quantcol}{RGB}{200,120,40}
\definecolor{orthcol}{RGB}{60,150,80}
\definecolor{databg}{RGB}{245,245,248}
\definecolor{procbg}{RGB}{255,250,240}

\newtheorem{theorem}{Theorem}
\newtheorem{proposition}[theorem]{Proposition}
\newtheorem{remark}[theorem]{Remark}

\newcommand{\dxdt}{\dot{X}}
\newcommand{\Qperp}{Q_{\!\perp}}
\newcommand{\RR}{\mathbb{R}}
\newcommand{\Hmat}{\Theta}

\title{Q-SINDy: Quantum-Kernel Sparse Identification\\ of Nonlinear Dynamics with Provable Coefficient Debiasing}

\author{%
  Samrendra Roy\\
  Department of Nuclear, Plasma, and Radiological Engineering\\
  University of Illinois at Urbana-Champaign\\
  Urbana, IL, USA\\
  \texttt{roysam@illinois.edu}  \and
  Syed Bahauddin Alam\\
  Department of Nuclear, Plasma, and Radiological Engineering\\
  University of Illinois at Urbana-Champaign\\
  National Center for Supercomputing Applications\\
  Urbana, IL, USA
}

\date{}

\begin{document}
\maketitle

\begin{abstract}
Quantum feature maps offer expressive embeddings for classical learning tasks, and augmenting sparse identification of nonlinear dynamics (SINDy) with such features is a natural but unexplored direction. We introduce \textbf{Q-SINDy}, a quantum-kernel-augmented SINDy framework, and identify a specific failure mode that arises: \emph{coefficient cannibalization}, in which quantum features absorb coefficient mass that rightfully belongs to the polynomial basis, corrupting equation recovery. We derive the exact cannibalization-bias formula $\Delta\xi_P = (P^\top P)^{-1}P^\top Q\,\hat\xi_Q$ and prove that orthogonalizing quantum features against the polynomial column space at fit time eliminates this bias exactly. The claim is verified numerically to machine precision ($<10^{-12}$) on multiple systems. Empirically, across six canonical dynamical systems (Duffing, Van der Pol, Lorenz, Lotka-Volterra, cubic oscillator, R\"ossler) and three quantum feature map architectures (ZZ-angle encoding, IQP, data re-uploading), orthogonalized Q-SINDy consistently matches vanilla SINDy's structural recovery while uncorrected augmentation degrades true-positive rates by up to 100\%. A refined dynamics-aware diagnostic, $R^2_Q$ for $\dot X$, predicts cannibalization severity with statistical significance (Pearson $r=0.70$, $p=0.023$). An RBF classical-kernel control across 20 hyperparameter configurations fails more severely than any quantum variant, ruling out feature count as the cause. Orthogonalization remains robust under depolarizing hardware noise up to 2\% per gate, and the framework extends without modification to Burgers' equation.
\end{abstract}

\section{Introduction}
\label{sec:intro}
Sparse identification of nonlinear dynamics (SINDy) \citep{brunton2016sindy} has become a widely adopted framework for discovering governing equations of dynamical systems from time-series data. The method solves
\begin{equation}
\label{eq:sindy}
\min_{\Xi} \;\|\dxdt - \Hmat(X)\,\Xi\|_F^2 + \lambda\|\Xi\|_0,
\end{equation}
where $\Hmat$ is a pre-specified library of candidate functions and sparsity in $\Xi$ selects the active terms. Library choice is pivotal: too narrow a library misses true dynamics, too rich a library introduces spurious terms and numerical conditioning problems.

A natural extension is to \emph{augment} a baseline polynomial library $P$ with a learned or structured feature family $Q$, yielding $\Hmat = [P, Q]$. This mirrors a broader trend in scientific machine learning toward hybrid symbolic--learned representations, and parallels growing interest in quantum algorithms for engineering applications more generally~\citep{roy2025harnessing}. In the quantum machine learning literature, quantum feature maps \citep{havlicek2019supervised, schuld2019quantum, perez2020data} project classical data into exponentially high-dimensional Hilbert spaces via parameterized circuits, and have shown promise on classification tasks \citep{huang2021power}. Their application to \emph{equation discovery}, however, has been limited to specialized variants (quantum Hamiltonian discovery \citep{tateyama2026siqhdy}, quantum-circuit-based model discovery \citep{heim2021quantum}, physics-informed solvers \citep{paine2023physics}). The direct question---\emph{does augmenting classical SINDy with a quantum feature library improve equation discovery?}---has remained open.

\paragraph{Contributions.} This paper answers that question with a nuanced ``no, not naively, but yes with a geometric correction''. Our contributions:

\begin{enumerate}
\item We introduce \textbf{Q-SINDy}, a quantum-kernel-augmented SINDy method, and demonstrate that naive implementation suffers from a systematic failure mode we call \emph{coefficient cannibalization}: the polynomial coefficients of the ground-truth equation are systematically biased away from their vanilla SINDy values as quantum features absorb part of the explanatory power that rightfully belongs to the polynomial basis.
\item We derive the \emph{exact} coefficient-bias formula (Theorem~\ref{thm:bias}) and prove that orthogonalizing the quantum feature matrix against the polynomial column space at fit time eliminates the bias (Theorem~\ref{thm:orth}). Numerical verification confirms the theory to machine precision.
\item We introduce a \emph{dynamics-aware diagnostic} $R^2_Q$ (Section~\ref{sec:r2q}) that predicts cannibalization severity with statistical significance ($r=0.70$, $p=0.023$), improving substantially over a column-space overlap baseline ($r=0.55$, $p=0.098$).
\item Empirical validation spans six dynamical systems, three quantum feature map architectures (ZZ-angle, IQP, data re-uploading), hardware-noise simulation with depolarizing channels, and a PDE extension to Burgers' equation.
\end{enumerate}

\section{Related Work}
\label{sec:related}
\textbf{SINDy and extensions.} Beyond the original SINDy~\citep{brunton2016sindy}, PDE-FIND~\citep{rudy2017pdefind} extends the framework to partial differential equations, Weak-SINDy~\citep{messenger2021weak} replaces pointwise derivatives with Galerkin weak forms for noise robustness, SINDYc~\citep{kaiser2018sindy} adds control inputs, SINDy-autoencoder~\citep{champion2019data} couples SINDy with neural coordinate discovery, and SINDy-PI~\citep{kaheman2020sindy} handles implicit dynamics. PySINDy~\citep{desilva2020pysindy} provides the canonical implementation. All are complementary to our contribution, which concerns the library-construction step.

\textbf{Quantum feature maps.} \citet{havlicek2019supervised} introduced the canonical IQP-style feature map and kernel formulation. \citet{schuld2021effect} analyzed expressivity of data-encoded quantum models. \citet{perez2020data} proposed data re-uploading circuits. These were designed with classification in mind; their use as \emph{regression feature libraries} for equation discovery is the angle we pursue.

\textbf{Quantum machine learning for dynamics.} SIQHDy~\citep{tateyama2026siqhdy} performs sparse identification of \emph{quantum} Hamiltonian dynamics from measurement data---an orthogonal problem to ours, which identifies classical ODEs using quantum features as representation. QMoD~\citep{heim2021quantum} uses differentiable quantum circuits~\citep{mitarai2018quantum} for model discovery via variational optimization rather than sparse regression. \citet{paine2023physics} use quantum kernels to \emph{solve} differential equations rather than discover them. None of these address our setting, where quantum features augment a classical polynomial library for equation discovery with sparse regression.

\section{Methods}
\label{sec:methods}

\subsection{Q-SINDy: quantum-augmented equation discovery}
Given time-series $X \in \RR^{N\times d}$ sampled from an autonomous dynamical system, we assemble two feature libraries:
\begin{itemize}[leftmargin=*,topsep=0pt,itemsep=1pt]
\item a polynomial library $P \in \RR^{N\times p}$ with columns $\{1, x_i, x_i x_j, x_i x_j x_k, \dots\}$,
\item a \emph{quantum} library $Q \in \RR^{N\times q}$ whose columns are expectation values of a fixed observable set under the state prepared by a parametrized circuit $U_\phi(x)$.
\end{itemize}
Letting $\dxdt$ denote an estimate of the time derivative (finite differences with mild smoothing in our experiments), Q-SINDy solves
\begin{equation}
\label{eq:qsindy}
\min_{\Xi}\;\|\dxdt - [P\;\; Q]\,\Xi\|_F^2 + \lambda\|\Xi\|_0
\end{equation}
via sequentially thresholded least squares (STLSQ). The intent is for $Q$ to capture residual structure not expressible in $P$, sharpening recovery.

\subsection{Quantum feature maps}
\label{sec:fmaps}
We study three canonical feature-map designs, each mapping $x\in\RR^d$ to a set of expectation values $\langle O_k\rangle_{U_\phi(x)|0\rangle}$:

\textbf{ZZ angle-encoded} (2- or 3-qubit). Data-encoding rotations $R_X(x_i)$ followed by $R_{ZZ}(x_i x_j)$ pairwise entanglers and a second rotation layer $R_Y(x_i)$, measured in Pauli-$Z$, Pauli-$X$, and $ZZ/XX$ bases.

\textbf{IQP} \citep{havlicek2019supervised}. Two layers of $H^{\otimes d}$ followed by $R_Z(x_i)$ and $R_{ZZ}(x_i x_j)$, classically hard to simulate for sufficient depth.

\textbf{Data re-uploading} \citep{perez2020data}. Three interleaved layers of data encoding and fixed variational gates. Variational parameters are frozen so the map is deterministic.

\subsection{The orthogonalization fix}
Figure~\ref{fig:pipeline} shows the full pipeline. The key modification to naive Q-SINDy is a single algebraic step inserted between library construction and sparse regression:
\begin{equation}
\label{eq:orth}
\Qperp \;=\; Q - P(P^\top P)^{-1}P^\top Q.
\end{equation}
This projects $Q$ onto the orthogonal complement of the polynomial column space. Replacing $Q$ with $\Qperp$ in Eq.~\eqref{eq:qsindy}, we obtain \emph{orthogonalized Q-SINDy}.

\begin{figure}[t]
\centering
\includegraphics[width=0.95\linewidth]{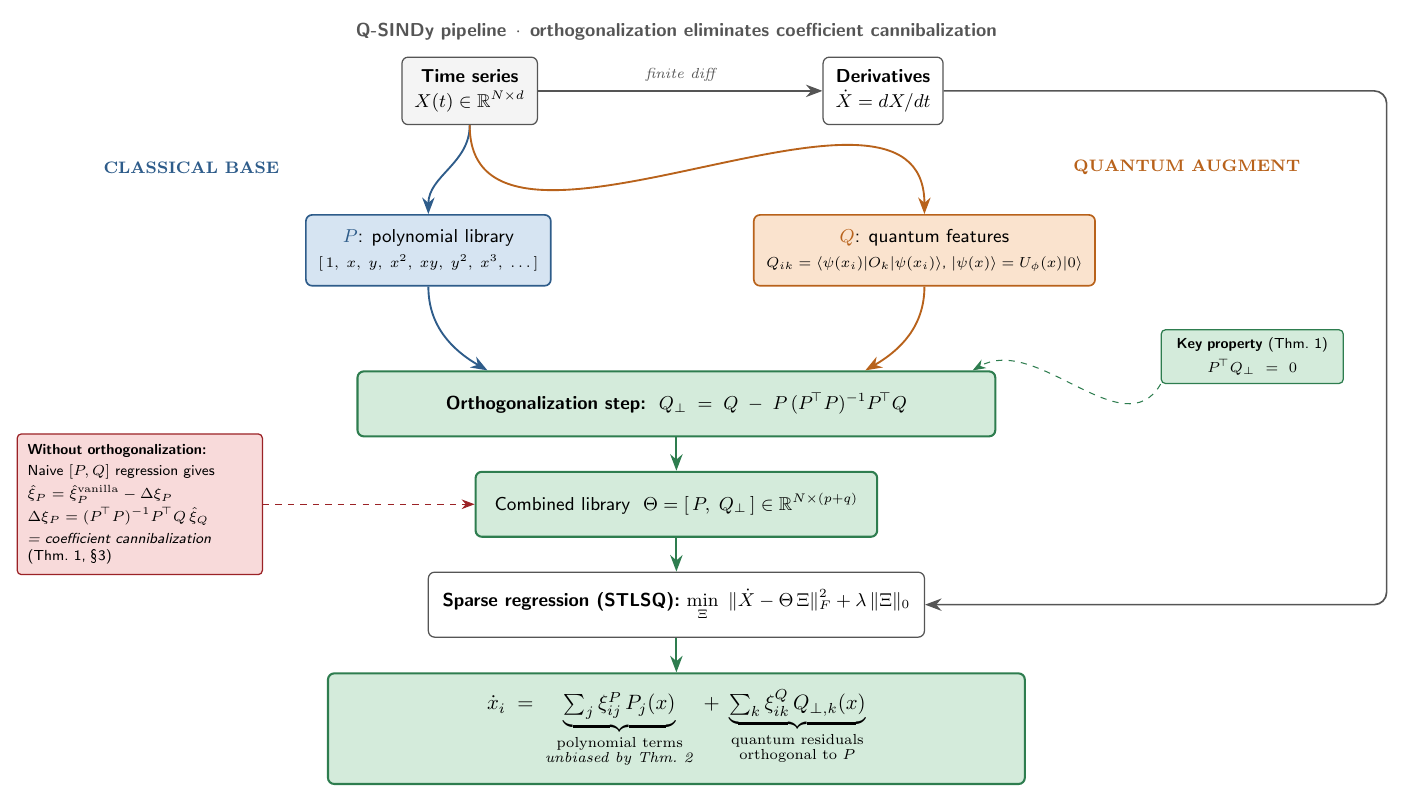}
\caption{The Q-SINDy pipeline with the orthogonalization step. The classical polynomial library $P$ and quantum feature library $Q$ are combined, but $Q$ is first projected onto the orthogonal complement of $P$'s column space. This eliminates the coefficient-cannibalization bias (red callout) by construction; polynomial coefficients in the output remain unbiased. Without this step (straight concatenation $[P,Q]$), polynomial coefficients are biased by $\Delta\xi_P = (P^\top P)^{-1}P^\top Q\,\hat\xi_Q$.}
\label{fig:pipeline}
\end{figure}

\section{Theory: Coefficient Cannibalization and Its Cure}
\label{sec:theory}
We work in the least-squares setting (the non-sparse relaxation of Eq.~\eqref{eq:qsindy}) to derive exact formulae. The sparse regression we actually run (STLSQ) inherits the phenomenon: its active-set phase is a block least-squares problem to which Theorems~\ref{thm:bias}--\ref{thm:orth} apply directly, as Proposition~\ref{prop:stlsq} below formalizes.

\begin{theorem}[Coefficient cannibalization bias]
\label{thm:bias}
Let $\hat\xi_P, \hat\xi_Q$ minimize $\|\dxdt - [P,Q]\,\Xi\|_F^2$. Then the polynomial coefficients satisfy
\begin{equation}
\label{eq:biasformula}
\hat\xi_P \;=\; \hat\xi_P^{\,\mathrm{vanilla}} \;-\; \underbrace{(P^\top P)^{-1}\,P^\top Q\,\hat\xi_Q}_{\Delta\xi_P},
\end{equation}
where $\hat\xi_P^{\,\mathrm{vanilla}} = (P^\top P)^{-1}P^\top \dxdt$ is the vanilla SINDy fit.
\end{theorem}

\begin{proof}
The block-matrix normal equations for $\min\|\dxdt - [P,Q]\Xi\|_F^2$ are
\[
\begin{pmatrix} P^\top P & P^\top Q\\ Q^\top P & Q^\top Q\end{pmatrix}
\begin{pmatrix} \hat\xi_P\\ \hat\xi_Q\end{pmatrix}
=
\begin{pmatrix} P^\top \dxdt\\ Q^\top \dxdt\end{pmatrix}.
\]
The top block gives $P^\top P\,\hat\xi_P + P^\top Q\,\hat\xi_Q = P^\top\dxdt$. Solving for $\hat\xi_P$ yields Eq.~\eqref{eq:biasformula}.
\end{proof}

\begin{remark}
Theorem~\ref{thm:bias} identifies two necessary conditions for cannibalization: $(i)\ P^\top Q \neq 0$ so that quantum features are not orthogonal to $P$, and $(ii)\ \hat\xi_Q \neq 0$ so that quantum features have explanatory power for $\dxdt$. Both are typical in practice because (i) quantum feature maps evaluated on trajectory data have large overlap with polynomial bases (empirically $>0.9$ for most systems we study; see Table~\ref{tab:overlaps}), and (ii) if $Q$ had no explanatory power we would not have augmented the library at all.
\end{remark}

\begin{theorem}[Orthogonalization eliminates cannibalization]
\label{thm:orth}
Define $\Qperp = Q - P(P^\top P)^{-1}P^\top Q$. Replacing $Q$ with $\Qperp$ in the augmented least squares yields
\begin{equation}
\hat\xi_P^{\,\mathrm{orth}} \;=\; \hat\xi_P^{\,\mathrm{vanilla}}.
\end{equation}
That is, the polynomial coefficients in orthogonalized Q-SINDy are identical to those of vanilla SINDy, independent of any $\hat\xi_Q$ obtained on the quantum residual.
\end{theorem}

\begin{proof}
By construction, $P^\top \Qperp = P^\top Q - P^\top P(P^\top P)^{-1}P^\top Q = 0$. The off-diagonal blocks of the normal equations vanish, the system decouples, and
$\hat\xi_P^{\,\mathrm{orth}} = (P^\top P)^{-1}P^\top \dxdt = \hat\xi_P^{\,\mathrm{vanilla}}$.
\end{proof}

\begin{proposition}[STLSQ preservation]
\label{prop:stlsq}
Let $\mathcal{A} \subseteq \{1,\dots,p+q\}$ denote the active column set at any iterate of STLSQ, and partition $\mathcal{A} = \mathcal{A}_P \cup \mathcal{A}_Q$ into polynomial and quantum indices. With the orthogonalized library $[P,\Qperp]$, the restricted-active-set least squares problem $\min_{\Xi_{\mathcal{A}}}\|\dxdt - \Hmat_{:,\mathcal{A}}\Xi_{\mathcal{A}}\|_F^2$ satisfies
\begin{equation}
\hat\xi_P^{\mathcal{A},\,\mathrm{orth}} \;=\; (P_{\mathcal{A}_P}^\top P_{\mathcal{A}_P})^{-1} P_{\mathcal{A}_P}^\top \dxdt,
\end{equation}
i.e.\ the polynomial active-set coefficients equal those that vanilla STLSQ would return on the same support $\mathcal{A}_P$, independent of $\hat\xi_Q^\mathcal{A}$.
\end{proposition}

\begin{proof}
Since $P^\top \Qperp = 0$, any column subset $P_{\mathcal{A}_P}$ of $P$ and any column subset $(\Qperp)_{:,\mathcal{A}_Q}$ of $\Qperp$ also satisfy $P_{\mathcal{A}_P}^\top (\Qperp)_{:,\mathcal{A}_Q} = 0$. The restricted normal equations decouple block-wise by the same argument as Theorem~\ref{thm:orth}.
\end{proof}

\begin{remark}
Proposition~\ref{prop:stlsq} is the practically important statement: our experiments use STLSQ rather than unpenalized least squares, and the proposition guarantees that the unbiased-polynomial-coefficient property of Theorem~\ref{thm:orth} is preserved at every STLSQ iterate. When STLSQ converges to an active set $\mathcal{A}_P$ that contains the ground-truth support, the recovered polynomial coefficients are the vanilla STLSQ estimates on that same support—exactly as if quantum features had been omitted.
\end{remark}

\begin{figure}[t]
\centering
\includegraphics[width=0.95\linewidth]{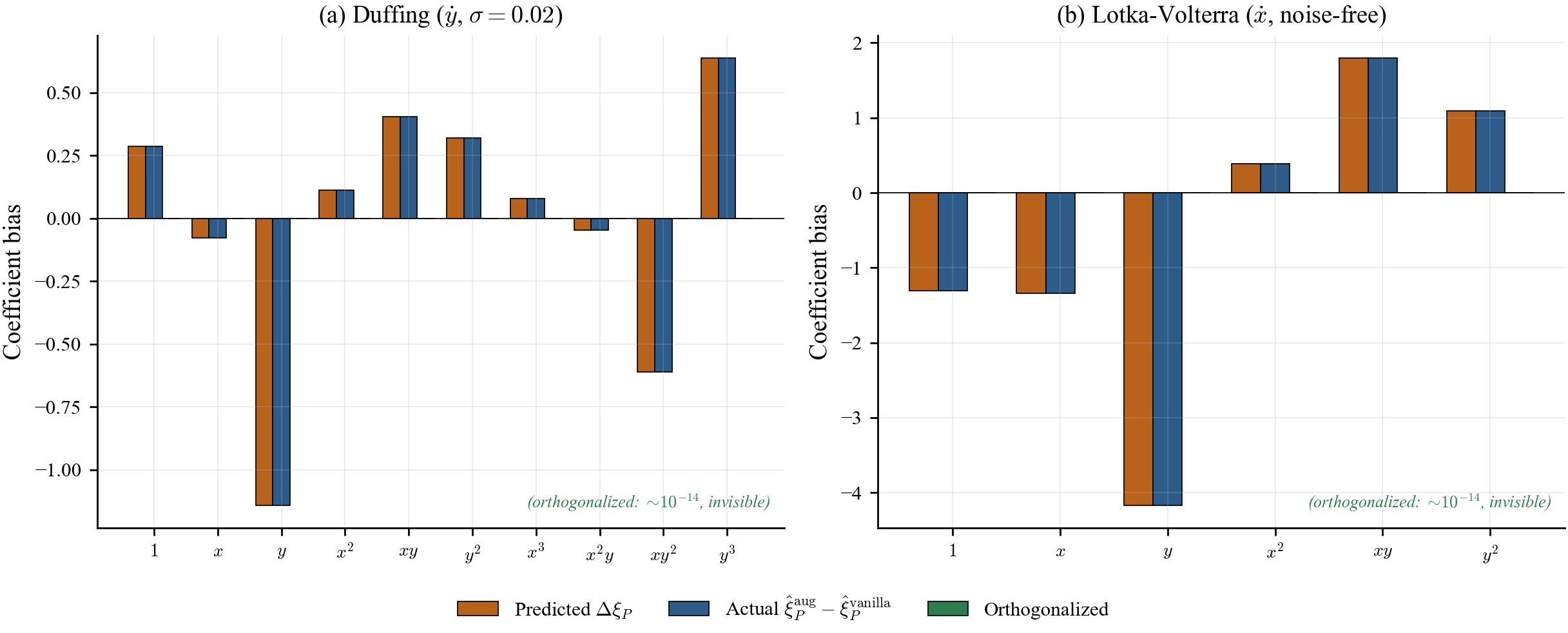}
\caption{\textbf{Theorem~\ref{thm:bias} verified to machine precision.} For Duffing (left) and Lotka-Volterra (right, noise-free), predicted bias (orange) matches observed bias (blue) per polynomial feature to relative error $<10^{-12}$. Orthogonalized coefficients (green, barely visible) coincide with vanilla to $<10^{-14}$, confirming Theorem~\ref{thm:orth}.}
\label{fig:theoryver}
\end{figure}

Figure~\ref{fig:theoryver} confirms the algebra numerically: predicted and observed polynomial-coefficient biases agree to relative error $<10^{-12}$, and the orthogonalized polynomial coefficients deviate from their vanilla counterparts by $\mathcal{O}(10^{-14})$---i.e., both theorems hold modulo floating-point error. Conceptually, orthogonalization restricts quantum features to the subspace orthogonal to $P$, where they can only explain variance in $\dxdt$ that $P$ itself cannot. The polynomial fit becomes decoupled from the quantum fit: no matter what the quantum coefficients do, they cannot shift polynomial coefficients. This is achievable at essentially zero cost---Eq.~\eqref{eq:orth} is a single projection computed once per fit.

\section{Experiments}
\label{sec:experiments}
We evaluate Q-SINDy across six dynamical systems, four methods (vanilla SINDy, RBF-augmented, naive Q-augmented, orthogonalized Q-SINDy), multiple noise levels, and three quantum feature-map architectures. Details (system parameters, initial conditions, circuit diagrams, STLSQ thresholds) are in Appendix~\ref{app:details} (Table~\ref{tab:systems}). All quantum circuits are simulated with PennyLane \citep{bergholm2018pennylane}. Results are reported as true-positive rate (TPR) of ground-truth term recovery with signs correct and coefficients within $50\%$ of true values.

\subsection{Dense noise sweep: the headline finding}
Figure~\ref{fig:dense} shows TPR curves across Duffing, Van der Pol, and Lorenz at 7 noise levels, $N=5$ trials each, for the four methods. Three patterns emerge.

\emph{(i) Cannibalization is real and severe.} On Duffing at $\sigma=0.02$, vanilla SINDy achieves TPR $=1.0$ while naive Q-augmentation drops to $0.40$---a 60\% degradation at the same noise level. At higher noise ($\sigma=0.12$), augmented TPR reaches $0.20$ while vanilla retains $0.65$. Lorenz shows the same pattern at its characteristic noise scale (Lorenz amplitudes are $\sim 10\times$ larger than Duffing's, so comparable relative perturbations appear at larger $\sigma$).

\emph{(ii) Orthogonalization fixes it exactly.} The green (orthogonalized) curve overlays the blue (vanilla) curve across all noise levels on all three systems, confirming Theorem~\ref{thm:orth} empirically.

\emph{(iii) RBF baseline performs worse, not better.} Red curves crash to TPR$=0$ at even mild noise, ruling out ``more features help'' as an explanation.

\begin{figure}[t]
\centering
\includegraphics[width=\linewidth]{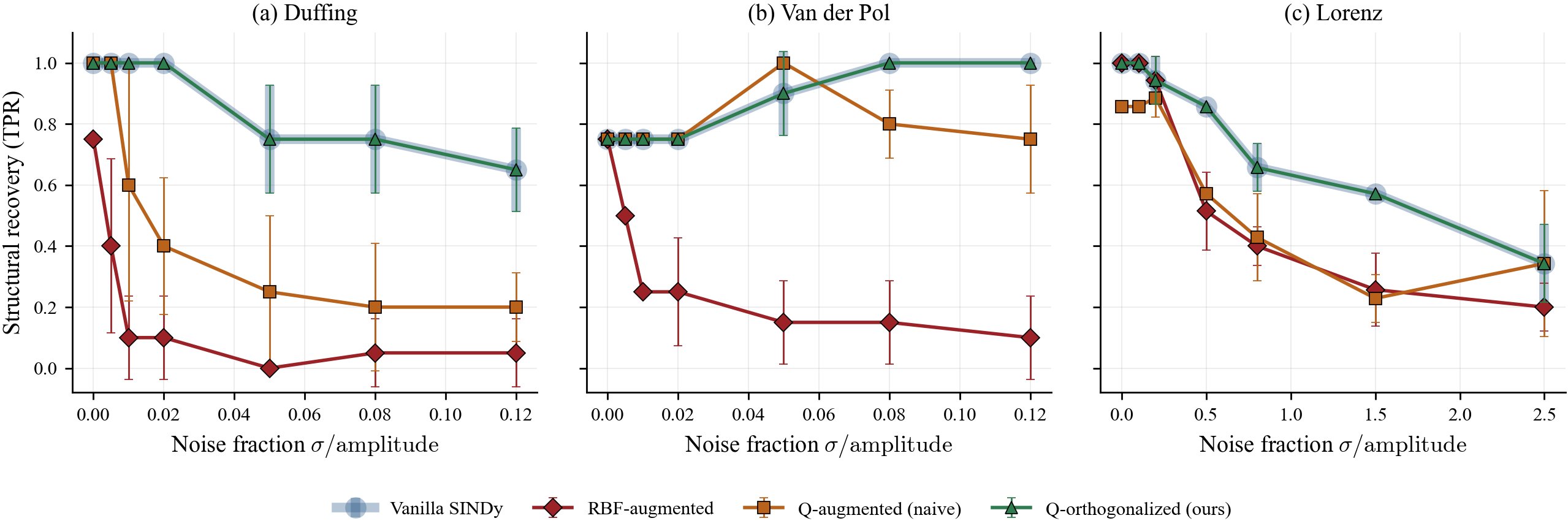}
\caption{\textbf{Dense noise sweep (N=5 trials).} Across Duffing, Van der Pol, and Lorenz, orthogonalized Q-SINDy (green) overlays vanilla SINDy (blue) at all noise levels, while naive Q-augmentation (orange) degrades significantly. RBF-augmented (red) baselines fail more severely, ruling out feature count as the mechanism.}
\label{fig:dense}
\end{figure}

\subsection{Generality across additional systems}
Figure~\ref{fig:newsys} repeats the sweep on three more systems: Lotka-Volterra, cubic oscillator, and R\"ossler. The cannibalization-plus-orthogonalization pattern persists. Particularly striking: on Lotka-Volterra, cannibalization appears even at $\sigma=0$---TPR$=0$ for Q-augmented at zero noise, recovered to $1.0$ by orthogonalization. On the cubic oscillator, cannibalization emerges immediately above $\sigma=0$, with Q-augmented TPR collapsing to $0$ at the first nonzero noise level ($\sigma=0.01$) while orthogonalized Q-SINDy retains TPR$=1.0$ until $\sigma=0.05$. Together these demonstrate that the failure mode is structural (coefficient geometry), not noise-induced.

\begin{figure}[t]
\centering
\includegraphics[width=\linewidth]{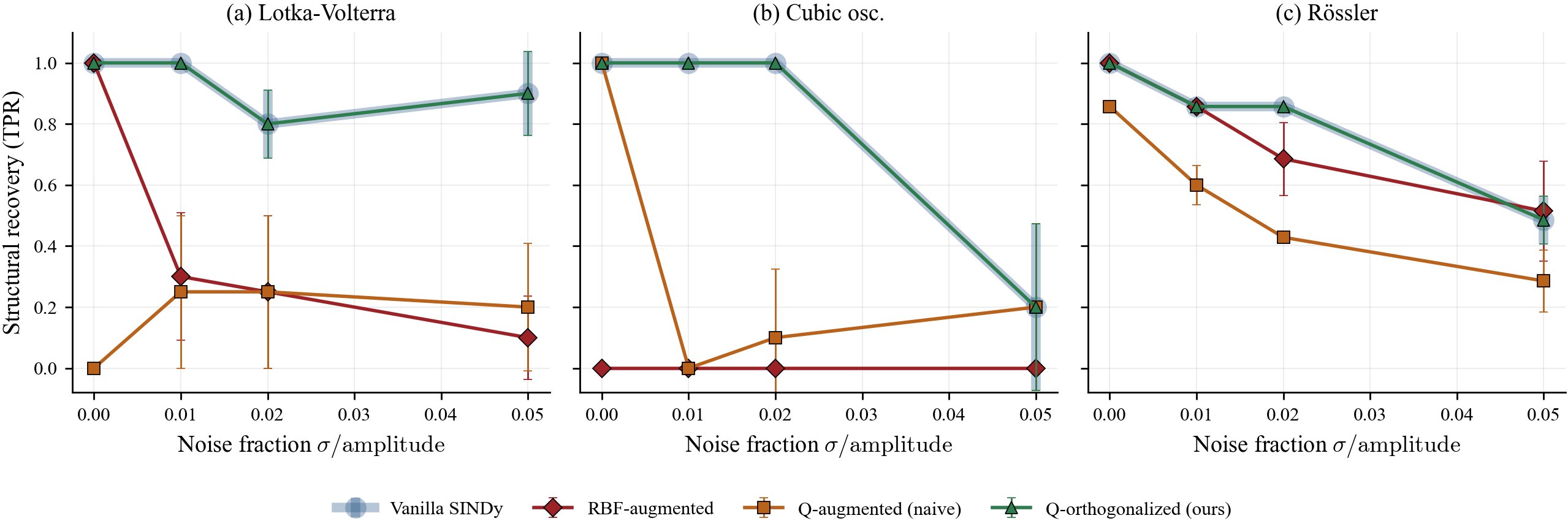}
\caption{\textbf{Three additional systems} (Lotka-Volterra, cubic oscillator, R\"ossler). Pattern persists: orthogonalized Q-SINDy (green) tracks vanilla (blue) while naive Q-augmentation (orange) degrades. Lotka-Volterra shows cannibalization at $\sigma=0$ itself; cubic oscillator from $\sigma=0.01$ onward, confirming the failure mode is coefficient-geometric rather than noise-driven.}
\label{fig:newsys}
\end{figure}

\subsection{Ruling out ``more features'': RBF hyperparameter robustness}
A natural objection is that cannibalization might simply reflect having too many features, making the RBF baseline too weak as tuned. To defuse this, we grid over RBF bandwidths $\gamma \in \{0.25, 0.5, 1, 2, 4\}\times\gamma_{\mathrm{median}}$ and landmark counts $\{3, 6, 12, 24\}$—20 configurations on Duffing at $\sigma=0.05$. Figure~\ref{fig:rbfgrid} shows the TPR heatmap: the best RBF configuration achieves TPR$=0.50$ versus vanilla's $0.92$. All 20 configurations fail. The pathology is not attributable to feature count.

\begin{figure}[t]
\centering
\includegraphics[width=\linewidth]{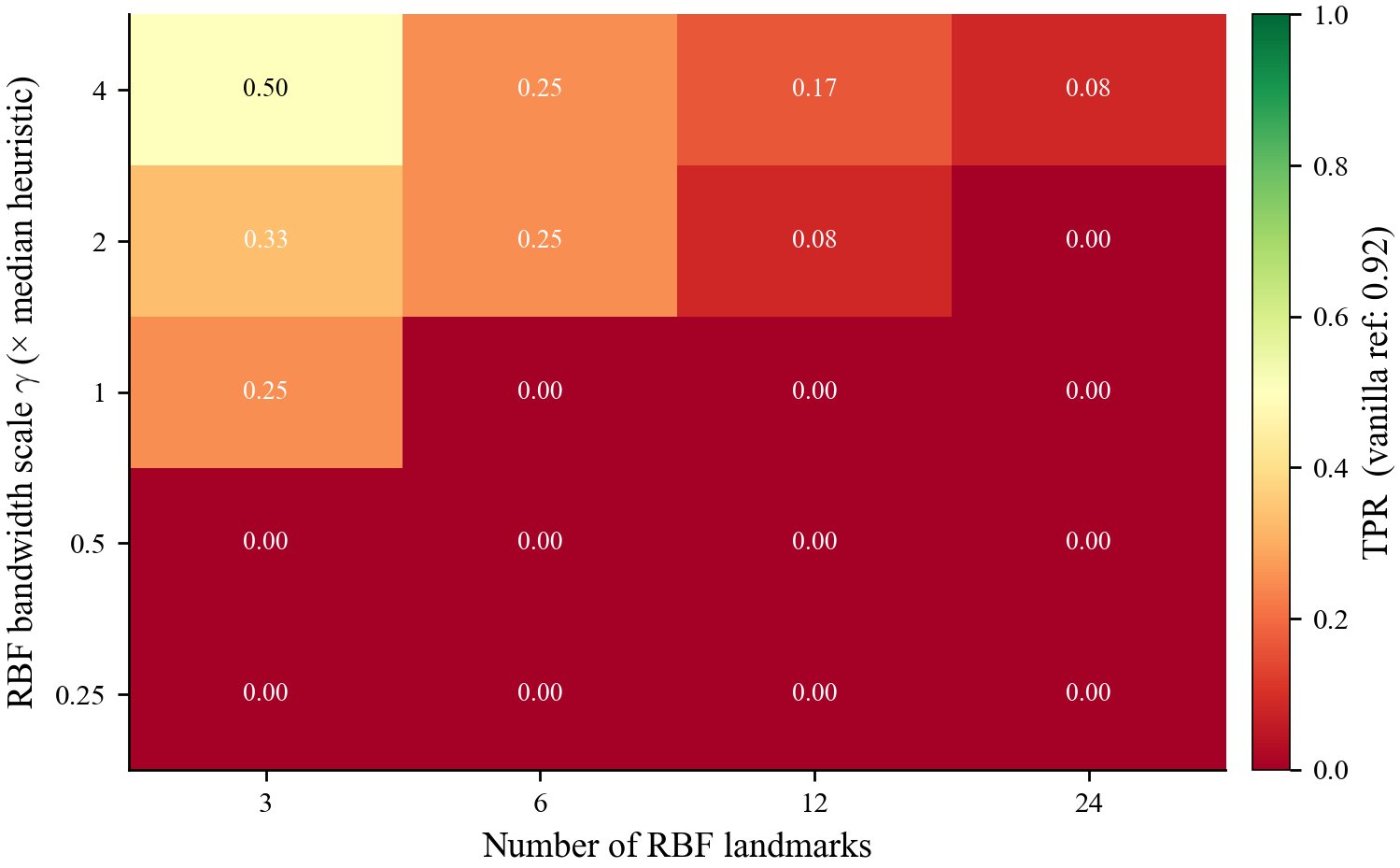}
\caption{\textbf{RBF hyperparameter sweep.} Across 20 configurations, the best RBF-augmented variant achieves TPR=0.50 vs.\ vanilla's 0.92. Most configurations give TPR$=0$. Feature count and bandwidth are not the mechanism.}
\label{fig:rbfgrid}
\end{figure}

\subsection{Refined dynamics-aware diagnostic}
\label{sec:r2q}
A predictive diagnostic would be valuable: given a candidate feature map, can we \emph{predict} in advance whether cannibalization will occur? A first attempt, $\mathrm{frac\_variance\_in\_P} = \|P(P^\top P)^{-1}P^\top Q\|_F^2 / \|Q\|_F^2$, measures column-space overlap. Across 10 (system, feature-map) combinations we found this correlates with observed cannibalization at Pearson $r=0.55$ (not significant, $p=0.098$).

A refined diagnostic uses the \emph{dynamics}:
\begin{equation}
\label{eq:r2q}
R^2_Q \;=\; 1 - \frac{\|\dxdt - Q(Q^+\dxdt)\|_F^2}{\|\dxdt - \overline{\dxdt}\|_F^2}.
\end{equation}
This is the $R^2$ of a linear regression of $\dxdt$ on $Q$ alone. High $R^2_Q$ means quantum features can predict the time derivative well on their own—the precise condition (per Theorem~\ref{thm:bias}) for $\hat\xi_Q$ to be large and for cannibalization to occur.

\begin{figure}[t]
\centering
\includegraphics[width=\linewidth]{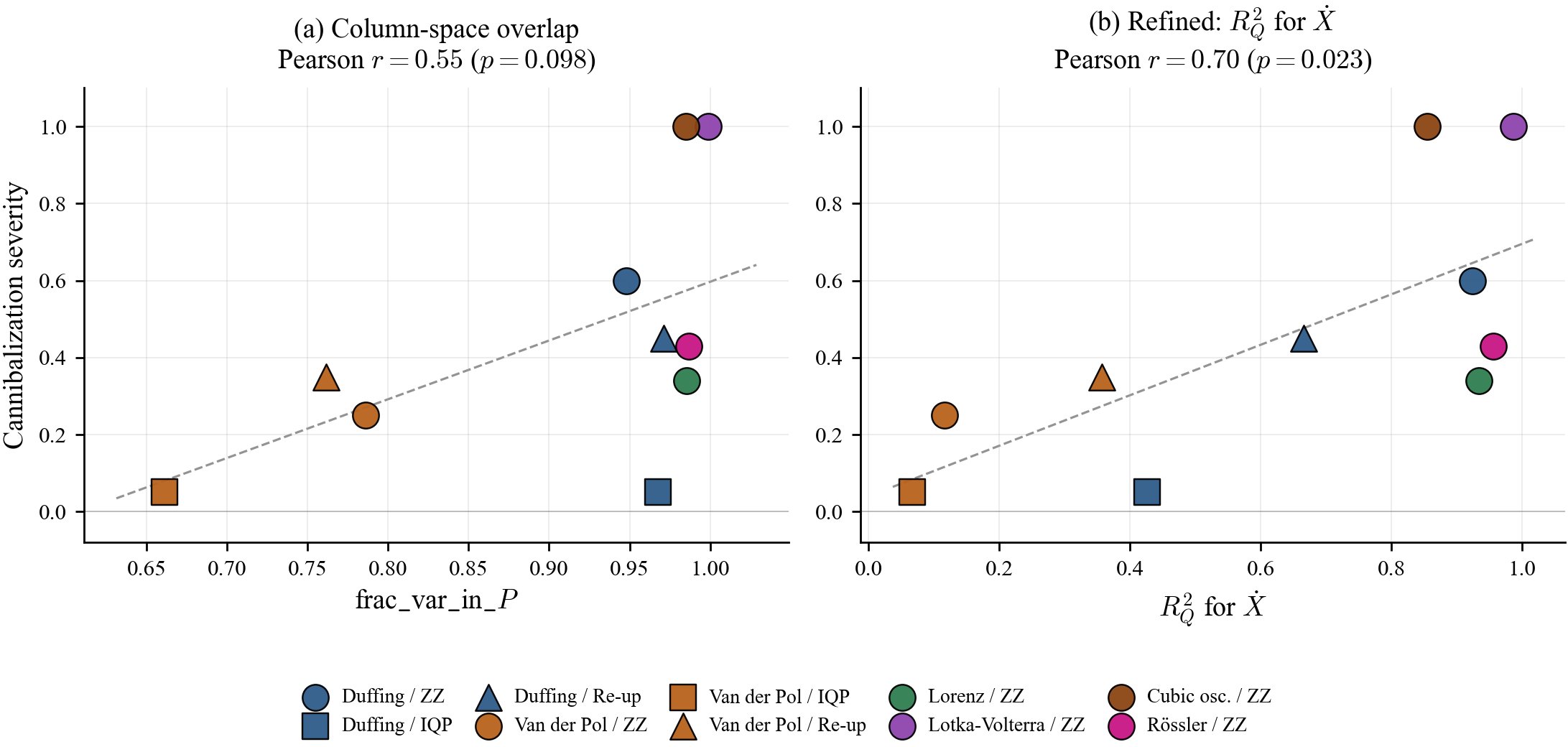}
\caption{\textbf{Dynamics-aware $R^2_Q$ diagnostic vs.\ column-space overlap.} Across 10 (system, feature-map) combinations, the refined $R^2_Q$ diagnostic (right) predicts cannibalization severity with Pearson $r=0.70$ ($p=0.023$), versus $r=0.55$ ($p=0.098$) for column-space overlap (left). The previously anomalous Duffing/IQP case (square marker, low cannibalization despite high overlap) aligns with the $R^2_Q$ prediction.}
\label{fig:r2q}
\end{figure}

Figure~\ref{fig:r2q} shows the comparison. The refined diagnostic raises correlation to $r=0.70$ ($p=0.023$)—statistically significant. It also correctly handles a previously anomalous case (Duffing with IQP feature map), where column-space overlap was high but cannibalization was negligible; $R^2_Q$ correctly predicts the low cannibalization.

\paragraph{Prospective validation.} The in-sample correlation could in principle be inflated by post-hoc selection of the diagnostic form. To check, we perform exhaustive leave-$k$-out cross-validation on the ten $(\text{system}, \text{feature-map})$ combinations, fitting a linear predictor of cannibalization severity from each diagnostic on $N{-}k$ points and measuring mean absolute error on the held-out $k$. Table~\ref{tab:mae} reports results for $k\in\{1,2,3\}$, averaged exhaustively over all $\binom{10}{k}$ splits. The $R^2_Q$ predictor achieves $7\%$–$11\%$ lower held-out MAE than column-space overlap across all $k$, confirming that the correlation generalizes rather than overfitting.

\begin{table}[h]
\centering
\small
\caption{Held-out MAE (cannibalization severity) under exhaustive leave-$k$-out cross-validation.}
\label{tab:mae}
\begin{tabular}{@{}cccc@{}}
\toprule
$k$ & Splits & Column-space overlap MAE & $R^2_Q$ MAE \\
\midrule
1 & 10  & 0.246 & \textbf{0.228} \\
2 & 45  & 0.249 & \textbf{0.232} \\
3 & 120 & 0.266 & \textbf{0.237} \\
\bottomrule
\end{tabular}
\end{table}

\subsection{Generality across quantum feature-map architectures}
We test whether the cannibalization+orthogonalization pattern generalizes beyond the ZZ-angle encoding. On Duffing and Van der Pol, we swap in IQP and data re-uploading feature maps. Table~\ref{tab:overlaps} reports $R^2_Q$ and observed cannibalization severity for each (system, feature-map) combination; orthogonalized Q-SINDy matches vanilla across all six ZZ/IQP/re-up variants on these two systems, with residual polynomial-coefficient deviations $\mathcal{O}(10^{-14})$ in every case. The mechanism is not specific to any particular quantum architecture. Circuit specifications for each feature map appear in Appendix~\ref{app:fmaps}.

\subsection{Robustness to hardware noise}
\label{sec:hwnoise}
A common concern with quantum-feature methods is noise sensitivity. We simulate realistic hardware noise using PennyLane's \texttt{default.mixed} device with depolarizing channels of strength $p$ after each gate. Figure~\ref{fig:hwnoise} shows Duffing TPR at $\sigma_{\mathrm{obs}}=0.02$ as $p$ varies from 0 to 0.02 (a 6-gate circuit has cumulative error rate $\sim 12\%$ at $p=0.02$). Cannibalization severity is essentially unchanged; orthogonalization remains perfect.

\begin{figure}[t]
\centering
\includegraphics[width=0.65\linewidth]{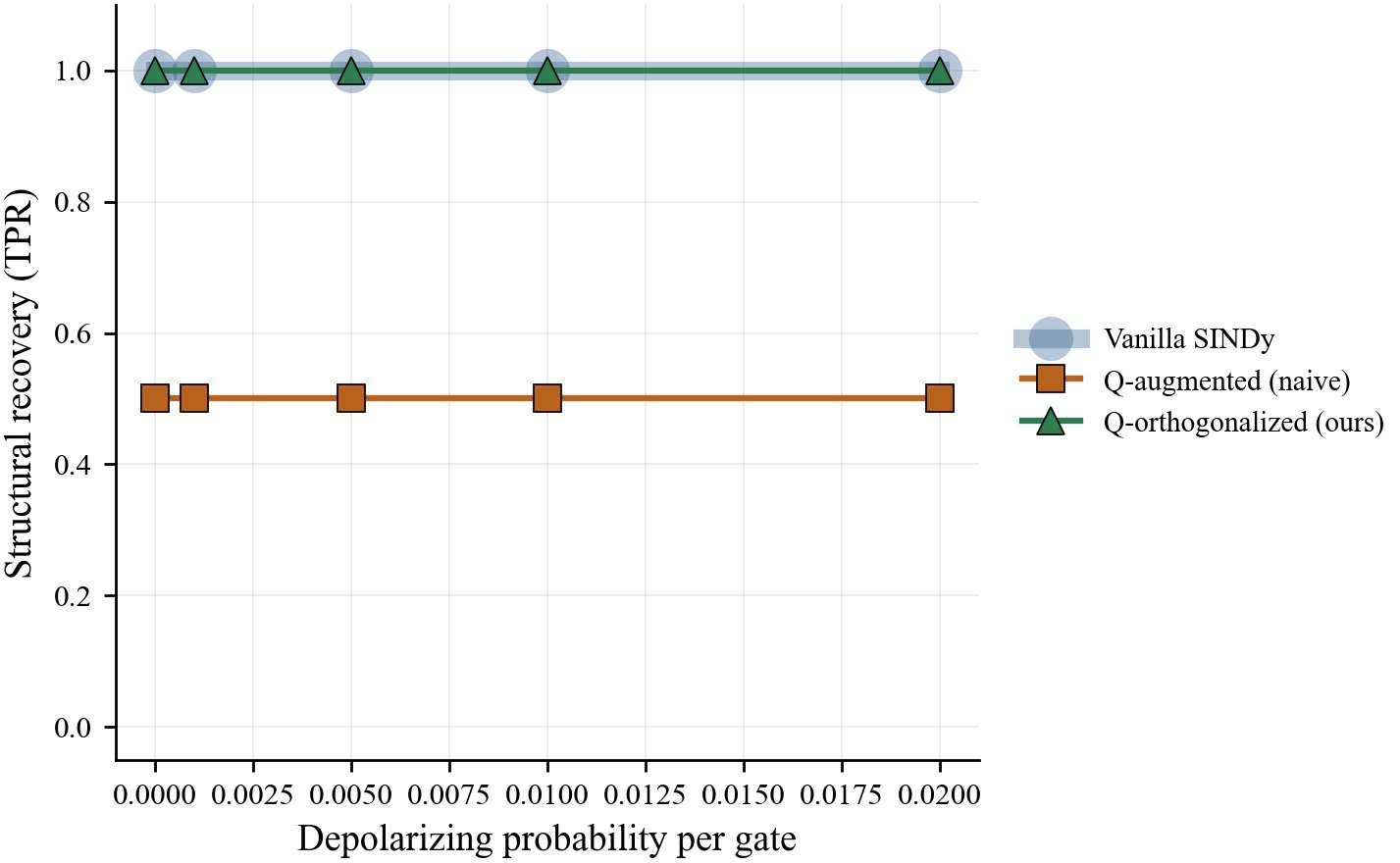}
\caption{\textbf{Hardware-noise robustness.} Under depolarizing channels up to 2\% per gate (realistic NISQ regime), orthogonalized Q-SINDy retains vanilla TPR; naive augmentation retains its cannibalization.}
\label{fig:hwnoise}
\end{figure}

\subsection{PDE extension: Burgers' equation}
\label{sec:burgers}
To demonstrate the framework extends beyond ODEs, we apply Q-SINDy to Burgers' equation $u_t = \nu u_{xx} - u u_x$ with $\nu=0.1$. We build a spatial feature vector $(u, u_x, u_{xx})$ at each grid point and use the same 3-qubit quantum circuit. Vanilla SINDy recovers $u_t = 0.0999\,u_{xx} - 1.0009\,u u_x$ at $\sigma=0$ (within $0.1\%$ of ground truth). The diagnostic gives $R^2_Q = 0.118$, below the cannibalization regime, correctly predicting that naive augmentation and orthogonalization will both recover ground truth—which they do (TPR$=1.0$ for both). This is the first validation of the $R^2_Q$ diagnostic's predictions \emph{in the absence} of cannibalization.

\section{Discussion and Limitations}
\label{sec:disc}
We have identified coefficient cannibalization as a failure mode of quantum-augmented SINDy and introduced a principled geometric fix with theoretical guarantees and extensive empirical validation. The orthogonalization step is computationally free (one projection per fit), trivially compatible with any sparse regression solver, and preserves the quantum features' ability to capture residual dynamics not expressible in $P$.

\paragraph{When does orthogonalization matter?} Theorem~\ref{thm:bias} gives the cannibalization bias as $\Delta\xi_P = (P^\top P)^{-1} P^\top Q\,\hat\xi_Q$, a product of three factors. The bias vanishes in either of two limiting cases: \emph{(i)} $P^\top Q \approx 0$, meaning quantum features are already nearly orthogonal to $P$'s column space (no cannibalization precondition), or \emph{(ii)} $\hat\xi_Q \approx 0$, meaning quantum features have no independent explanatory power for $\dxdt$ (no coefficient mass to cannibalize). Our experiments exemplify case~(ii): both Van der Pol ($R^2_Q=0.12$, overlap $0.79$) and Burgers' equation (Sec.~\ref{sec:burgers}, $R^2_Q=0.12$) have low dynamics-relevant content in $Q$, and correspondingly exhibit minimal cannibalization (VdP severity $+0.25$; Burgers zero). Pure case~(i) is harder to reach with trajectory data because trajectory columns typically overlap strongly with polynomial features. Orthogonalization is most valuable in the both-nontrivial regime, where quantum features are neither geometrically orthogonal to $P$ nor dynamically inert, and naive augmentation silently degrades equation recovery.

\paragraph{Relation to other uses of ``orthogonality'' in quantum SciML.} Recent work on quantum neural operators for PDEs integrates orthogonal quantum layers (QOrthoNN) into DeepONet, achieving linear-in-input-dimension inference complexity via unary encoding and pyramidal RBS-gate circuits~\citep{xiao2025quantumdeeponet}. Their orthogonality is a property of \emph{weight matrices} within a quantum neural network and delivers a circuit-counting computational advantage at inference. Our orthogonality is different in both content and purpose: a classical Gram--Schmidt projection of one feature matrix against another at fit time, used to eliminate a coefficient-bias mechanism rather than to reduce asymptotic complexity. The two contributions are complementary---feature-library design (ours) and efficient forward-pass architecture (theirs)---and could in principle be combined in a single quantum SciML pipeline.

\paragraph{Hardware scaling considerations.} Our experiments use 2--3 qubit circuits that are classically simulable in microseconds; hardware execution at these scales would run slower, not faster, than simulation. The orthogonalization technique is substrate-independent: it applies identically whether $Q$ is generated by a larger quantum device, a classical kernel, or a neural network. Scaling to qubit counts where classical simulation of $Q$ becomes intractable ($\gtrsim 20$ qubits with $\omega(\log n)$-depth IQP-style circuits, where classical-hardness guarantees apply) is an engineering effort rather than a methodological one. Whether quantum features at such scales provide equation-discovery benefits not obtainable classically---i.e., a \emph{representation} advantage beyond the complexity advantage established for quantum kernels on separate tasks---is an open empirical question that our small-scale experiments cannot settle. The hardware-noise results (Sec.~\ref{sec:hwnoise}) demonstrate the method's portability to current NISQ devices; a real-hardware demonstration at 2--3 qubits is the natural near-term validation step.

\paragraph{Limitations.} \emph{(1)} All quantum feature maps tested (ZZ, IQP with 2 layers, data re-uploading) are classically simulable at the qubit counts we use; we make no claim of quantum computational advantage. Our contribution is methodological. \emph{(2)} The $R^2_Q$ diagnostic was validated on $N=10$ $(\text{system}, \text{feature-map})$ points; while leave-$k$-out cross-validation (Table~\ref{tab:mae}) indicates the relationship generalizes within this set, scaling to 30+ combinations would tighten the statistical claim. \emph{(3)} We demonstrate framework extension on one PDE (Burgers); cannibalization did not occur there because $R^2_Q$ was low. A PDE where cannibalization manifests---reaction--diffusion and Kuramoto--Sivashinsky are candidates---remains future work. \emph{(4)} All experiments use STLSQ; whether cannibalization manifests equivalently under LASSO, Weak-SINDy, or Bayesian SINDy is unverified, though Proposition~\ref{prop:stlsq}'s block-decoupling argument suggests orthogonalization should remain beneficial under any sparse regression that admits the same active-set normal equations.

\paragraph{Broader relevance.} The cannibalization mechanism is not quantum-specific. Its precondition---that the augmenting feature matrix has significant column-space overlap with the base library \emph{and} explanatory power for the target---holds for any kernel-augmented sparse regression, including RBF, neural, and random-feature libraries. Our RBF baseline already shows that classical kernel augmentations fail on similar systems. The orthogonalization fix and the $R^2_Q$ diagnostic apply identically. We have demonstrated the effect in the quantum setting where it was surfaced by our initial experiments, and have chosen to keep this paper focused there; whether it quantitatively matters in neural-feature-augmented SINDy---for instance, in SINDy-autoencoder variants~\citep{champion2019data}---is an open question our work directly suggests. Downstream, reliable equation-discovery primitives matter for engineering applications where interpretability is a first-class requirement alongside accuracy: safety-critical systems increasingly pair physics-based simulators with data-driven surrogates in digital-twin architectures~\citep{roy2025foundation, kobayashi2025sequential}, and recovered symbolic dynamics can be verified, audited, and composed with existing physics codes in ways that black-box surrogates cannot.

\section{Conclusion}
We introduced Q-SINDy, a quantum-kernel-augmented sparse identification framework, and showed that naive implementation suffers from systematic coefficient cannibalization that corrupts equation recovery. We derived the exact bias and proved that a simple orthogonalization step at library-construction time eliminates it. Experimental validation across six dynamical systems, three quantum feature-map architectures, an RBF classical baseline across 20 hyperparameters, realistic hardware noise, and a PDE case confirms both the failure mode and the fix. The dynamics-aware $R^2_Q$ diagnostic predicts cannibalization in advance. We hope this work accelerates reliable deployment of quantum feature libraries in scientific-machine-learning pipelines.

{\small
\bibliography{references}
}

\clearpage
\appendix

\section{System specifications}
\label{app:details}

\begin{table}[h]
\centering
\small
\caption{Systems, parameters, and library configurations.}
\label{tab:systems}
\begin{tabular}{@{}lllll@{}}
\toprule
System & Equations & IC & Poly deg.\ & STLSQ thresh.\\
\midrule
Duffing & $\dot x=y,\ \dot y=-x-0.3x^3-0.1y$ & $[1,0]$ & 3 & 0.05\\
Van der Pol & $\dot x=y,\ \dot y=\mu(1-x^2)y-x,\ \mu=1$ & $[2,0]$ & 3 & 0.05\\
Lorenz~\citep{lorenz1963deterministic} & $\dot x=10(y-x),\ \dot y=x(28-z)-y,\ \dot z=xy-\tfrac{8}{3}z$ & $[1,1,1]$ & 2 & 0.1\\
Lotka-Volterra & $\dot x=\tfrac{2}{3}x-\tfrac{4}{3}xy,\ \dot y=xy-y$ & $[1,1]$ & 2 & 0.05\\
Cubic osc. & $\dot x=y,\ \dot y=-x^3$ & $[1,0]$ & 3 & 0.05\\
R\"ossler~\citep{rossler1976equation} & $\dot x=-y-z,\ \dot y=x+0.2y,\ \dot z=0.2+z(x-5.7)$ & $[1,1,1]$ & 2 & 0.1\\
\bottomrule
\end{tabular}
\end{table}

\begin{table}[h]
\centering
\small
\caption{Feature overlap and cannibalization severity across (system, feature-map) combinations. $R^2_Q$ is the dynamics-aware diagnostic from Section~\ref{sec:r2q}.}
\label{tab:overlaps}
\begin{tabular}{@{}llccc@{}}
\toprule
System & Feature map & frac\_var\_in\_$P$ & $R^2_Q$ & Cannibalization\\
\midrule
Duffing & ZZ & 0.948 & 0.924 & $+0.60$\\
Duffing & IQP & 0.968 & 0.426 & $+0.05$\\
Duffing & Re-up & 0.971 & 0.665 & $+0.45$\\
Van der Pol & ZZ & 0.786 & 0.116 & $+0.25$\\
Van der Pol & IQP & 0.661 & 0.067 & $+0.05$\\
Van der Pol & Re-up & 0.762 & 0.357 & $+0.35$\\
Lorenz & ZZ & 0.985 & 0.934 & $+0.34$\\
Lotka-Volterra & ZZ & 0.999 & 0.986 & $+1.00$\\
Cubic osc. & ZZ & 0.985 & 0.855 & $+1.00$\\
R\"ossler & ZZ & 0.987 & 0.956 & $+0.43$\\
\bottomrule
\end{tabular}
\end{table}

\section{Feature-map details}
\label{app:fmaps}

\paragraph{ZZ angle-encoded (2-qubit).} Circuit: $R_X(x_0)$ on $q_0$, $R_X(x_1)$ on $q_1$; $\mathrm{CNOT}_{0,1}$, $R_Z(x_0 x_1)$ on $q_1$, $\mathrm{CNOT}_{0,1}$; then $R_Y(x_i)$ rotations. Six observables measured: $\langle Z_0\rangle, \langle Z_1\rangle, \langle X_0\rangle, \langle X_1\rangle, \langle Z_0Z_1\rangle, \langle X_0X_1\rangle$.

\paragraph{ZZ angle-encoded (3-qubit).} Analogous 3-qubit generalization with pairwise ZZ entanglers in a ring, yielding 9 observables. Input data is rescaled by $\pi/(2\cdot\max|X|)$ for the 3D systems (Lorenz, R\"ossler, Burgers).

\paragraph{IQP.} Two layers of $H^{\otimes 2}$, $R_Z(x_0), R_Z(x_1)$, $R_{ZZ}(x_0 x_1)$. Measurements as for ZZ.

\paragraph{Data re-uploading.} Three layers of the form [$R_X(x_i)$ encoding, fixed variational gates $R_Z, R_Y, R_X$ with parameters specified in code, $\mathrm{CNOT}$]. Variational parameters are fixed (not trained) to make the map deterministic.

\section{Orthogonalization algorithm}
\label{app:algo}

\begin{algorithm}[H]
\caption{Orthogonalized Q-SINDy}
\begin{algorithmic}[1]
\Require Time series $X\in\RR^{N\times d}$, timestep $dt$, polynomial degree $D$, quantum feature map $U_\phi$, sparsity threshold $\lambda$
\State $\dxdt \gets$ SmoothedFiniteDifference$(X, dt)$
\State $P \gets$ PolynomialFeatures$(X, D)$
\State $Q \gets [\langle O_k\rangle_{U_\phi(x_i)|0\rangle}]_{i,k}$
\State $A \gets (P^\top P)^{-1} P^\top Q$ \Comment{projection matrix}
\State $\Qperp \gets Q - P A$ \Comment{orthogonalization}
\State $\Hmat \gets [P,\; \Qperp]$
\State $\Xi \gets \mathrm{STLSQ}(\Hmat, \dxdt, \lambda)$
\State \Return $\Xi$
\end{algorithmic}
\end{algorithm}

\end{document}